\documentclass[11pt]{amsart}
\pdfoutput=1
\usepackage{amsmath}
\usepackage{mathtools}
\usepackage{tikz}
\usetikzlibrary{arrows}
\usepackage{graphicx}
\usepackage{caption}
\usepackage{subcaption}
\usepackage{multicol}

\definecolor{sepia}{HTML}{704214}

\usepackage{array}
\newcolumntype{C}{>{\centering\arraybackslash}m{8pt}}
\usepackage{multirow}
\newtheorem{theorem}{Theorem}
\newtheorem{lemma}[theorem]{Lemma}
\newtheorem{corollary}[theorem]{Corallary}
\newtheorem{gadget}[theorem]{Gadget}
\newcommand*\cov[1]{\overline{#1}}

\title[Complexity of SCSN \& SCSNW]{Approximation hardness of Shortest Common Superstring variants}
\author{Y. William Yu}
\address[Y. William Yu] {Department of Mathematics\\
    Massachusetts Institute of Technology\\
Cambridge, Massachusetts 02139}
\email[Y.W.~Yu]{ywy@mit.edu}
\thanks{Research supported by a Hertz Foundation Fellowship.}

\date{\today}


\begin{document}
\begin{abstract}
    The shortest common superstring (SCS) problem has been studied at great length because of its connections to the \textit{de novo} assembly problem in computational genomics. The base problem is APX-complete, but several generalizations of the problem have also been studied. In particular, previous results include that SCS with Negative strings (SCSN) is in Log-APX (though there is no known hardness result) and SCS with Wildcards (SCSW) is Poly-APX-hard. Here, we prove two new hardness results: (1) SCSN is Log-APX-hard (and therefore Log-APX-complete) by a reduction from Minimum Set Cover and (2) SCS with Negative strings and Wildcards (SCSNW) is NPOPB-hard by a reduction from Minimum Ones 3SAT.
\end{abstract}
\maketitle

\section{Introduction}
Given a set of strings $s_1,\ldots, s_n$, the Shortest Common Superstring optimization problem (SCS) is to minimize $N$ so that there exists a string $S$ of length $N$ such that all $s_i$ are substrings of $S$.
SCS and its variants are closely related to the assembly problem in computational genomics \cite{medvedev2007computability}; i.e. piecing together a full genome from small fragments, though redundancy in the genome implies that the correspondence is not perfect. Note that this not to be confused with Shortest Common Supersequence, which deals with subsequences instead of substrings, and can be related to the alignment problem in genomics \cite{raiha1981shortest}.

SCS was proven in 1994 to be APX-hard from the traveling salesman problem (TSP) using a $2n+1$ length alphabet, and a 3-approximation algorithm was given \cite{blum1994linear}.
Since then, a string of algorithmic advances have brought the approximation ratio down to $\frac{57}{23}$ \cite{mucha2013lyndon,golovnev2013approximating}, and inapproximability results have shown that the minimium is $\frac{333}{332}$ \cite{karpinski2013improved}.
Additionally, Ott showed in 1999 that SCS is APX-hard even when the alphabet is restricted to size 2 by a reduction from TSP with all distances either 1 or 2 \cite{ott1999lower}.

Although base SCS is thus fairly well-characterized as APX-complete, several generalizations have also been studied in the literature (Table \ref{tab:prevknown}).
In particular, allowing for negative strings (which are not allowed in the superstring) and wildcards seem to increase the difficulty of the problem.
Shortest Common Superstring with Negative strings (SCSN) can be approximated to within a logarithmic factor \cite{jiang1994approximating} using the Group-Merge algorithm \cite{li1990towards}, but no comparable hardness result has been shown in the literature.
Shortest Common Superstring with Wildcards (SCSW) on the other hand is known to be Poly-APX-hard by reduction from minimum chromatic number \cite{ma2009greed}.
Nothing is known about the combination of the two, Shortest Common Superstring with Negative strings and Wildcards (SCSNW).

\begin{table}[tbp]
    \begin{center}
        \begin{tabular}{| c | c | c | }
            \hline
            &   \textbf{No Negative Strings}    & \textbf{Negative Strings} \\ \hline
            \textbf{No Wildcards} &   SCS: APX-complete    &   SCSN: in Log-APX \\ \hline
            \textbf{Wildcards} &     SCSW: Poly-APX-hard   &   SCSNW: {\color{red}???} \\ \hline
        \end{tabular}
    \caption{Previously known results for variants of SCS.}
    \label{tab:prevknown}
    \end{center}
\end{table}

\begin{table}
    \centering
    \begin{tabular}{| c | c | c | }
        \hline
        &   \textbf{No Negative Strings}    & \textbf{Negative Strings} \\ \hline
        \textbf{No Wildcards} &   SCS: APX-complete    &  SCSN: {\color{blue}Log-APX-complete} \\ \hline
        \textbf{Wildcards} &     SCSW: Poly-APX-hard   &  SCSNW: {\color{blue}NPOPB-hard} \\ \hline
    \end{tabular}
    \caption{Updated table of approximability results for variants of SCS with new results presented in this paper in blue.}
    \label{tab:nowknown}
\end{table}

In this paper, we first briefly review existing reductions for proving APX-hardness and Poly-APX-hardness of SCS and SCSW respectively. Building on insights and strategies from those reductions, we then prove two new hardness results: (1) SCSN is Log-APX-hard and (2) SCSNW is NPOPB-hard.

\section{Reductions Review}
In this section we review reductions for SCS and SCSW. We omit many details, as we are interested only in highlighting some of the gadgets and reduction strategies that we will be using later.

\subsection{SCS reduction from $O(1)$-degree vertex cover \cite{vassilevska2005explicit}}
Given a set of strings $s_1,\ldots, s_n$, SCS wants to minimize $N$ so that there exists a string $S$ of length $N$ such that all $s_i$ are substrings of $S$.
Although the original APX-hardness reduction for SCS was from a variant of the Traveling Salesperson Problem \cite{blum1994linear}, we review here (in brief, skipping many details) a more recent reduction from $O(1)$-degree vertex cover \cite{vassilevska2005explicit}, as our new reductions build on several of the ideas.

We start with instance of vertex cover $G = (V,E)$ with $|V| = n$ and $|E| = m$.
Let the alphabet $\Sigma = V$ so each vertex $a$ is associated with a single letter $a$.
Let an edge $(a,b)$ be represented by strings $abab$ and $baba$.
Suppose $G$ has a vertex set $S$ of size $k$. Assign every edge $(a,b)$ to its covering vertex (or arbitrarily if both vertices are in $S$). Then if $a$ is the assigned vertex for the edge $(a,b)$, overlap the two strings to get $ababa$, else overlap the other way to get $babab$.
Then for every $c \in S$, we can overlap all assigned edge strings by 1 to get $ca_1ca_1ca_2ca_2c\ldots ca_{k_c}ca_{k_c}c$ of length $4k_c +1$, where $k_c$ is the number of edges assigned to $c \in S$.
By concatenating all such strings together, we get a superstring of length $4m+k$ (Figure \ref{fig:scs}).

Conversely, it can be shown that all superstrings for the SCS problem are of length $4m + t$, and can be shortened in polynomial time to a string corresponding to a vertex cover. Thus, if we can get a superstring of length $4m + k$ for SCS, we can get a vertex cover of size $\le k$. Making use of exact bounds from the $O(1)$-vertex cover problem, it is possible to show SCS is APX-hard.

We will reuse two of the gadgets later in the SCSN Log-APX-hardness proof:
\begin{enumerate}
    \item Overlapping strings in two different ways for each edge to select which vertex covers that edge.
    \item Creating vertex strings by overlapping edge strings, such that each additional vertex used contributes 1 to the final cost.
\end{enumerate}

\begin{figure}[tb]
    \centering
    \begin{subfigure}[b]{0.39\textwidth}
        \begin{tikzpicture}[->,>=stealth',shorten >=1pt,auto,node distance=3cm,
                            thick,main node/.style={circle,draw,font=\sffamily\Large\bfseries}]

          \node[main node] (1) {1};
          \node[main node, fill=blue!30] (2) [below left of=1] {2};
          \node[main node] (3) [below right of=2] {3};
          \node[main node, fill=blue!30] (4) [below right of=1] {4};

          \path[every node/.style={font=\sffamily\small}]
            (1) edge node [left] {1414} (4)
                edge node[left] {1212} (2)
            (2) edge node [right] {2121} (1)
                edge node {2424} (4)
                edge node[left] {2323} (3)
            (3) edge node [right] {3232} (2)
                edge node[right] {3434} (4)
            (4) edge node [left] {4343} (3)
                edge node[right] {4141} (1)
                edge node {4242} (2);
        \end{tikzpicture}
    \end{subfigure}
    \begin{subfigure}[b]{0.6\textwidth}
        \small
        \begin{itemize}
            \item If $S = \{2,4\}$, then edges collapse to 
                \begin{itemize}
                    \item 21212
                    \item 23232
                    \item 24242
                    \item 41414
                    \item 43434
                \end{itemize}
            \item Then the two vertices $2$ and $4$ are associated with strings
                \begin{itemize}
                    \item 2121232324242
                    \item 414143434
                \end{itemize}
            \item Which results in a final string
                \begin{itemize}
                    \item 2121242423232 414143434
                \end{itemize}
        \end{itemize}
    \end{subfigure}
    \caption{SCS reduction from $O(1)$-degree vertex cover example.}
    \label{fig:scs}
\end{figure}

\subsection{SCSW reduction from minimum chromatic number \cite{ma2009greed}}
Given set of strings $s_1, \ldots, s_n$ with letters from $\Sigma \cup \{?\}$ find the shortest string $S$ with letters from $\Sigma$ that is a superstring of all $s_i$, where each $?$ can match any letter of $\Sigma$.
For genomics, this corresponds to uncertainty in sequencer calls for particular bases in a DNA read.

We start with a minimum chromatic number problem on graph $G = (V,E)$ with $V = \{v_1, v_2, \ldots, v_n\}$ and $E = \{e_1, e_2, \ldots, e_m\}$.
Let $\Sigma = \{A, T, G, C\}$. For each $v_i$, let $t_i$ be a string of length $m$ such that 
        \[ t_i[k] = \left\{
        \begin{array}{l l}
            A, & \textrm{if $e_k = (v_i, v_j)$ and $i < j$} \\
            T, & \textrm{if $e_k = (v_j, v_i)$ and $j < i$} \\
            ?, & \textrm{otherwise.} 
        \end{array}  \right.   \]
Then let $s_i = Xt_iX, \hspace{1em} \forall i \in [1,n]$ where $X = G^{mn}C^{mn}$, be the SCSW instance.
By construction, independent sets can completely overlap with one another. As each color in a coloring corresponds to an independent set, superstrings have length proportional to the minimum chromatic number (exactly $2mn + m (2n +1)k$, where $k$ is the chromatic number).
Any superstring of the SCSW problem can be polynomially shortened to be of the form $XY_1XY_2X\ldots XY_k$.
Reconstructing the independent sets is then just matter of reading off the set edges in each string between the $X$ border markers.

Unlike the SCS reduction in the last section, this is an L-reduction \cite{crescenzi1997short} (or more precisely, after normalizing by $m(2n+1)$, it is an L-reduction). This is because each new color needed in min chromatic number corresponds to not just a single character, but instead an entire substring $XY_iX$'s worth.

We will reuse two of the strategies later:
\begin{enumerate}
    \item Using wildcards to allow collapsing together many input strings into a single section.
    \item Forcing each additional color to correspond to a long string so that we have an L-reduction.
\end{enumerate}
While the first strategy is only applicable to our SCSNW proof later, the second is used in both the SCSN and SCSNW reductions in the next section.

\begin{figure}[tb]
    \centering
    \begin{subfigure}[b]{0.39\textwidth}
        \begin{tikzpicture}[->,>=stealth',shorten >=1pt,auto,node distance=3cm,
                            thick,main node/.style={circle,draw,font=\sffamily\Large\bfseries}]

          \node[main node, fill=blue!30] (1) {1};
          \node[main node, fill=red!30] (2) [below left of=1] {2};
          \node[main node, fill=blue!30] (3) [below right of=2] {3};
          \node[main node, fill=yellow!30] (4) [below right of=1] {4};

          \path[every node/.style={font=\sffamily\small}]
            (1) edge node [left] {2} (4)
                edge node[left] {1} (2)
            (2) edge node [right] {} (1)
                edge node {4} (4)
                edge node[left] {3} (3)
            (3) edge node [right] {} (2)
                edge node[right] {5} (4)
            (4) edge node [left] {} (3)
                edge node[right] {} (1)
                edge node {} (2);
        \end{tikzpicture}
    \end{subfigure}
    \begin{subfigure}[b]{0.6\textwidth}
        \begin{itemize}
            \item Strings in SCSW are 
                \begin{itemize}
                    \item $s_1$ = \texttt{{\color{green}X}{\color{blue}AA???}{\color{green}X}}
                    \item $s_2$ = \texttt{{\color{green}X}{\color{red}T?AA?}{\color{green}X}}
                    \item $s_3$ = \texttt{{\color{green}X}{\color{blue}??T?A}{\color{green}X}}
                    \item $s_4$ = \texttt{{\color{green}X}{\color{sepia}?T?TT}{\color{green}X}}
                \end{itemize}
            \item Vertices $s_1$ and $s_3$ can merge due to independence
                \begin{itemize}
                    \item \texttt{{\color{green}X}{\color{blue}AAT?A}{\color{green}X}}
                \end{itemize}
            \item Which results in a final string
                \begin{itemize}
                    \item \texttt{{\color{green}X}{\color{blue}AAT?A}{\color{green}X}{\color{red}T?AA?}{\color{green}X}{\color{sepia}?T?TT}{\color{green}X}}
                \end{itemize}
        \end{itemize}
    \end{subfigure}
    \caption{SCSW reduction from minimum chromatic number example.}
    \label{fig:scsw}
\end{figure}

\section{New hardness results}
\subsection{SCSN reduction from minimum set cover}
\begin{theorem}[SCSN is Log-APX-hard]
    Given a set of strings $s_1, \ldots, s_\eta$ and a set of negative strings $t_1, \ldots, t_{poly(\eta)}$, both built from an alphabet $\Sigma$, optimizing for the shortest string $T$ that is a superstring of all $s_i$ but contains no $t_j$ as a substring is Log-APX-hard.
    \label{thm:SCSN}
\end{theorem}
We will prove this theorem by reduction from min set cover \cite{lund1994hardness}, but will first need some setup.
Our strategy for this reduction will be to use the negative strings to force certain structural conditions.
Let the set cover problem be to cover the set of items $S = \{1, \ldots, m\}$ by sets $\cov{1}, \ldots, \cov{n} \in C$, $\cov{i} \subset S$.
Let the alphabet $\Sigma = S \cup C \cup \{l_0, l, b, e\}$.
We introduce the additional letters $b, e$ (begin and end) to frame the string to remove border effects, and the additional letters $l_0, l$ to force \textit{long} gaps after certain patterns.
For the reduction, let the input positive strings be $\{i \in S\} \cup \{b, e\}$, so we require that each item letter appear at least once and have a particular beginning and end.

\begin{lemma}
    For any string $X \in T$, we can disallow arbitrary prefixes and suffixes of bounded length $AXB$ in polynomial time.
    \label{lem:dispresuf}
\end{lemma}
\begin{proof}
    The total number of possible strings of the form $AXB$ is $|\Sigma|^{|A| + |B|}$, and listing them all out as negative strings takes $O((|A|+|X|+|B|)\cdot|\Sigma|^{|A| +|B|}$ time, which is polynomial.
    For ease of notation, we will use ``$?$'' as a ``wildcard'' symbol where applicable.
\end{proof}

\begin{gadget}[Frame Gadget]
    Forces $T = b?\cdots?e$.
    \label{gad:frame}
\end{gadget}
\begin{proof}[Design]
Disallow $?b$ and $e?$. Then there can be no letters left of $b$ or right of $e$ because then there would be a disallowed substring.
For the remainder of this section, unless explicitly noted otherwise, we will not consider $b$ and $e$ valid characters for substrings, since they must be unique and at predetermined locations.
\end{proof}

\begin{lemma}
    For any string $X \in T$, prefix length $u$, and suffix length $v$, we can force $X$ to extend to a string $AXB$ with $|A|=u, |B|=v$ such that $AXB$ is drawn from a specified set $T(X)$ in polynomial time.
\end{lemma}
\begin{proof}
    By Lemma \ref{lem:dispresuf}, we can list all strings of the form $?X?$ as negative strings in polynomial time.
    First, we list all strings of the form $?X?$ except those that match some center in $T(X)$ as negative strings.
    Then we iterate, building single characters onto prefix and suffix until we reach strings of form $AXB$.
    This takes polynomial time provided $u$ and $v$ as $u$ and $v$ are bounded.
    Additionally, because of the Frame Gadget, the iteration cannot stop until we reach $AXB$ because otherwise some other character would be the left-most or right-most in $T$.
\end{proof}

\begin{gadget}[Item Gadget]
    Extends the item string ``$i$'' to 
    \[
    \cov{j} i \cov{j} \cov{(j+1)} i \cov{(j+1)} \ldots \cov{n} i \cov{n} \cov{1} i \cov{1} \cov{2} i \cov{2} \ldots \cov{j} i \cov{j},
    \]
    for any choice of rotation $j \in \{1, \ldots, n\}$ for which $i \in \cov{j}$
\end{gadget}
\begin{proof}[Design]
Extend $?i?$ to $\cov{c}i\cov{c}$ where $\cov{c} \in C$. This forces every item to be surrounded on both sides by one of the sets in a triple.

For every $\cov{j} i \cov{j}$, disallow $??? \cov{j} i \cov{j} ???$ except
\begin{align*}
    \cov{(j-1)} i \cov{(j-1)} \cov{j} i \cov{j} ??? \hspace{1em} \textrm{or} \hspace{1em} ??? \cov{j} i \cov{j} \cov{(j+1)} i \cov{(j+1)} & \hspace{2em} \textrm{if } i \in \cov{j} \\
    \cov{(j-1)} i \cov{(j-1)} \cov{j} i \cov{j} \cov{(j+1)} i \cov{(j+1)} & \hspace{2em} \textrm{if } i \not \in \cov{j} 
\end{align*}
This forces every triple to be connected on at least one side to its consecutive triple, and buries in the middle triples corresponding to sets that do not cover $i$.

Then, for every string of $p$ triples 
\[X = \cov{j} i \cov{j} \ldots \cov{(j+p-1 \mod n)} i \cov{(j+p-1 \mod n}),\] for $p \in \{1, \ldots n\}$, disallow $??? X ???$ except for
\begin{align*}
    ??? X \cov{(j+1)} X \cov{(j+1)} \hspace{1em} \textrm{or} \hspace{1em} \cov{(j-1)} i \cov{(j-1)} X ???
\end{align*}
This forces these item gadgets to have at least $n+1$ triples.
To make sure they do not have more than $n+1$ triples, we just disallow strings with $n+2$ triples.
Thus, we have constructed our length $3n+3$ Item Gadget.
\end{proof}

\begin{gadget}[Set Gadget]
    Allows a $2+m(3n+2)$ penalty to be placed on the string length for each additional set needed in the cover, resulting in an L-reduction.
\end{gadget}
\begin{proof}[Design]
    As with the earlier SCS reduction, note that every item can be assigned to a particular set $\cov{j}$ for the cover by rotating the Item Gadget so that it starts and ends with $\cov{j}$.
    Then, adjacent items assigned to the same set can overlap by 1 character, so for a set $c$ with $k(c)>0$ assigned items, the set gadget will use up $k(c)\cdot(3n+2)+1$ characters.
    Alone, using the same arguments as in the SCS reduction, this would imply that the superstring uses up $2+m(3n+2)+k$ characters, for a set cover of size $k$.

    Unfortunately, the above is not an L-reduction as like in the SCSW reduction we need a multiplicative penalty.
    However, we can achieve that by forcing additional space between adjacent set gadgets.
    To do this, for every orientation of an item gadget $X$, disallow $X?$ except for $Xy$, where $y \in S \cup \{l_0, e\}$.
    Within a set gadget, the items overlap, so after an individual set gadget will be an item character, so this does not affect the internals of the set gadgets.
    However, after a set gadget, it must either be the end of the string $e$, or the character $l_0$.
    Now disallow $l_0 l^q ?$ except $l_0 l^q l$ for $q \in [0, m(3n+2)-1]$, forcing any substring starting with $l_0$ to have shape $l_0 l^{m(3n+2)}$ and thus length $1+m(3n+2)$.
    Thus, the space between adjacent set gadgets is thus $2+m(3n+2)$. For a set cover of size $k$, there are $k-1$ such spaces, so the final superstring will have length $k(2 + m(3n+2)) + 1$.
    By normalizing to $2 + m(3n+2)$, this implies that we have an L-reduction.
\end{proof}

\begin{proof}[Proof of Theorem \ref{thm:SCSN}]
    For any instance of min set cover, convert it to an instance of SCSN with alphabet size $n+m+2$ by the gadgets described in this section.
    As this is an L-reduction, and min set cover is Log-APX-complete, SCSN is Log-APX-hard for an alphabet of size $n+m+2$.
    However, by Theorem 1 in reference \cite{vassilevska2005explicit}, which proves that larger alphabet sizes can be encoded in polynomial time in a binary alphabet, SCSN is Log-APX-hard even for binary alphabets, showing that SCSN is Log-APX-hard for any alphabet, provided that the number of negative strings is allowed to be polynomial in the number of positive strings, completing the proof.
\end{proof}

\begin{corollary}
SCSN is Log-APX-complete.
\end{corollary}
\begin{proof}
Recall the existence of a log-approximation algorithm \cite{jiang1994approximating}.
Combined with Log-APX-hardness, this implies that SCSN is Log-APX-complete.
\end{proof}

\subsection{SCSNW reduction from minimum ones 3SAT}
\begin{theorem}[SCSNW is NPOPB-hard]
Given set of strings $s_1, \ldots, s_\eta$ and set of negative strings $t_1, \ldots, t_{poly(\eta)}$ with letters from $\Sigma \cup \{?\}$, optimizing for the shortest string $S$ with letters from $\Sigma$ that is a superstring of all $s_i$, but contains no $t_j$ as a substring, where each $?$ can match any letter of $\Sigma$, is NPOPB-hard.
\label{thm:SCSNW}
\end{theorem}
We will prove this theorem by reduction from min ones 3SAT (or Distinguished Ones 3SAT), which is NPOPB-complete \cite{kann1994polynomially}.
Our strategy will be to use a frame gadget to force all clause gadgets to overlap a particular section of the string consisting of variables that can be set true or false.
Then, using the variable gadget, we force each variable set to true to cause a large penalty by pushing a substing onto the end of the superstring.
In the following, we assign positive strings by ``$s_i =$'' and negative strings by ``$t_i =$''.

Additionally, we choose here the alphabet $\Sigma = \{A, T, G, C\}$ to match the bases in the human genome.
In the following gadgets, we also use the notation $B = GC , X = C^n G^n, R = (GAC)^n$ for the sake clarity and brevity.

\begin{gadget}{Frame gadget}
    Forces clause gadgets to overlap and variable gadgets not in the variable region to not overlap.
\end{gadget}
\begin{proof}[Design]
    \begin{align*}
        \textrm{Frame gadget } & \left\{
        \begin{array}{r l }
            s_{f_1} &=   B   X \overbrace{ ? \cdots \cdots \cdots \cdots \cdots ?}^{n \textrm{ variable slots}} X  R X \\
            t_{f_2} &=  \textrm{$?$} B  \\
            \underbracket{t_{f_{ijk}}}_{\forall i \in [1,3n^2], j \in [i+1,i+n]} & = R \overbrace{ ? \cdots ? \underbracket{A}_{i} ? \cdots ? \underbracket{A}_{j} ? \cdots ? }^{3n^2+n \textrm{ potential variable slots}} \\
            \underbracket{t_{v_{i,C}}}_{\forall i \in [1,n]} & = BX \overbrace{ ? \cdots ? \underbracket{C}_{i} ? \cdots ? } ^{n \textrm{ variable slots}} X
            \\
            \underbracket{t_{v_{i,G}}}_{\forall i \in [1,n]} & = BX \overbrace{ ? \cdots ? \underbracket{G}_{i} ? \cdots ? } ^{n \textrm{ variable slots}} X \\
        \end{array} \right.
    \end{align*}
    The $s_{f_1}$ string specifies the locations for the $n$ variable set variables, and the $t_{v_i,C}$ and $t_{v_i,G}$ negative strings ensure that those variable locations are either $T$ or $A$.
    The $t_{f_2}$ negative string forces the superstring to start with $B$, constraining where strings can go in the superstring.
    The $t_{f_{ijk}}$ negative strings ensure that no two variable gadgets can overlap except through their respective $X$ strings if they are to the right of $R$.
\end{proof}

\begin{gadget}{Variable gadget}
    Forces any set variable corresponding to push this gadget out to the end of the superstring.
\end{gadget}
\begin{proof}[Design]
    \begin{align*}
        \textrm{Variable gadget } & \left\{
        \begin{array}{l l }
            \underbracket{s_{v_{i}}}_{\forall i \in [1,n]} = X \overbrace{ ? \cdots ? \underbracket{A}_{i} ? \cdots ? } ^{n \textrm{ variable slots}} X \\
        \end{array} \right. 
    \end{align*}
    If this variable gadget is in the variable region between $B$ and $R$ in the frame, then the corresponding variable must be set to false.
    However, if the corresponding variable is set to true, then the entire gadget must be pushed over to the right of $R$, and cannot overlap except maximally by overlapping their $X$ regions.
    Thus, each additional variable set to true costs an additional $3n$ characters to the length of the string.
\end{proof}

\begin{gadget}{Clause gadget}
    Requires the variable section of the superstring to be set matching the clauses in the min ones 3SAT problem.
\end{gadget}
\begin{proof}[Design]
    \begin{align*}
        \textrm{Clause gadget } & \left\{
        \begin{array}{l l }
            \underbracket{t_{c}}_{c = v_i \vee v_j \vee \neg v_k} = B X \overbrace{ ? \cdots ? \underbracket{A}_{i} ? \cdots ? \underbracket{A}_{j} ? \cdots ? \underbracket{T}_{k} ? \cdots ? } ^{n \textrm{ variable slots}} \\
            \textrm{(positions are $T$ if negated in clause and $A$ otherwise)}
        \end{array} \right.
    \end{align*}
    For each clause, we create a negative string with all positions set to the opposite of what we want in the variable region. Thus, we disallow having all variables being the opposite of what would be needed to satisfy the clause. Thus, at least one of the variables must satisfy the clause, so all clauses with these gadgets must be satisfied. This is what forces some of the variables to be set to true.
\end{proof}

\begin{proof}[Proof of Theorem \ref{thm:SCSNW}]
    For any instance of min ones 3SAT, convert it to an instance of SCSNW with alphabet size $4$ by the gadgets described in this section.
    Construction is polynomial and takes $O(n^5)$ operations, most of which are used up constructing the negative strings of the frame gadget.
    If there is a Min Ones solution of weight $W$, then the corresponding SCSNW problem has a solution string of length $2 + 2n + n + 2n + 3n + 3nW + n = 2 + 9n + 3nW$.
    For any solution to the SCSNW problem, one gets a solution to min ones 3SAT by simply reading off the variable locations, and the weight of that solution is no more than $W$ given a superstring of length $2 + 9n + 3nW$.
    Note that by omitting some of the variable strings, this reduction also works for minimum distinguished ones 3SAT.
    After normalizing the SCSNW objective by $n$ (or equivalently the length of the longest input string), this reduction is an L-reduction.
    As min ones 3SAT (or min distinguished ones 3SAT) is NPOPB-hard, so thus must be SCSNW, completing the proof.
\end{proof}

As an aside, one might attempt to apply this reduction to SCSN, given that Lemma \ref{lem:dispresuf} can be generalized to allow wildcards in arbitrary positions in SCSN.
That would of course lead to contradictory results given that SCSN is known to be in Log-APX, and would imply an error in this proof.
However, note that this proof of SCSNW hardness requires access to $poly(n)$ wildcards per string and the proof of Lemma \ref{lem:dispresuf} can only be generalized to allow a constant number of wildcards per string (otherwise, we would need an exponential number of negative strings in SCSN).
Thus, this reduction cannot be used for SCSN, and SCSNW is provably harder than SCSN.

\section{Discussion}
We reviewed the complexity of SCS and variants depending on whether negative strings and wildcards were allowed, and built on those proofs to get new hardness results: SCSN is Log-APX-hard and therefore Log-APX-complete and SCSNW is NPOPB-hard (Table \ref{tab:nowknown} in intro).
We conjecture that SCSNW is in NPOPB if there exists a feasible solution, which would imply NPOPB-completeness, but this is nontrivial to show.
Future work could include proving completeness results for SCSW and SCSNW. 

\section{Acknowledgments}
Y.W.Y. is supported by a Hertz Foundation fellowship.
The author thanks Erik Demaine, Jayson Lynch, Sarah Eisenstat, and the entire Fall 2014 MIT 6.890 class for insightful conversations. Sarah Eisenstat is especially acknowledged for finding prior results in the literature that the author missed.
This manuscript was originally conceived as a final project for the Fall 2014 MIT 6.890 class \textit{Algorithmic Lower Bounds: Fun with Hard Proofs}, taught by Erik Demaine.

\bibliographystyle{amsalpha}
\bibliography{main}

\newcommand{\etalchar}[1]{$^{#1}$}
\providecommand{\bysame}{\leavevmode\hbox to3em{\hrulefill}\thinspace}
\providecommand{\MR}{\relax\ifhmode\unskip\space\fi MR }
\providecommand{\MRhref}[2]{%
  \href{http://www.ams.org/mathscinet-getitem?mr=#1}{#2}
}
\providecommand{\href}[2]{#2}
\begin{thebibliography}{MGMB07}

\bibitem[BJL{\etalchar{+}}94]{blum1994linear}
Avrim Blum, Tao Jiang, Ming Li, John Tromp, and Mihalis Yannakakis,
  \emph{Linear approximation of shortest superstrings}, Journal of the ACM
  (JACM) \textbf{41} (1994), no.~4, 630--647.

\bibitem[Cre97]{crescenzi1997short}
Pierluigi Crescenzi, \emph{A short guide to approximation preserving
  reductions}, Computational Complexity, 1997. Proceedings., Twelfth Annual
  IEEE Conference on (Formerly: Structure in Complexity Theory Conference),
  IEEE, 1997, pp.~262--273.

\bibitem[GKM13]{golovnev2013approximating}
Alexander Golovnev, Alexander~S Kulikov, and Ivan Mihajlin, \emph{Approximating
  shortest superstring problem using de bruijn graphs}, Combinatorial Pattern
  Matching, Springer, 2013, pp.~120--129.

\bibitem[JL94]{jiang1994approximating}
Tao Jiang and Ming Li, \emph{Approximating shortest superstrings with
  constraints}, Theoretical Computer Science \textbf{134} (1994), no.~2,
  473--491.

\bibitem[Kan94]{kann1994polynomially}
Viggo Kann, \emph{Polynomially bounded minimization problems that are hard to
  approximate}, Nordic Journal of Computing \textbf{1} (1994), 317--331.

\bibitem[KS13]{karpinski2013improved}
Marek Karpinski and Richard Schmied, \emph{Improved inapproximability results
  for the shortest superstring and related problems}, Proceedings of the
  Nineteenth Computing: The Australasian Theory Symposium-Volume 141,
  Australian Computer Society, Inc., 2013, pp.~27--36.

\bibitem[Li90]{li1990towards}
Ming Li, \emph{Towards a dna sequencing theory (learning a string)},
  Foundations of Computer Science, 1990. Proceedings., 31st Annual Symposium
  on, IEEE, 1990, pp.~125--134.

\bibitem[LY94]{lund1994hardness}
Carsten Lund and Mihalis Yannakakis, \emph{On the hardness of approximating
  minimization problems}, Journal of the ACM (JACM) \textbf{41} (1994), no.~5,
  960--981.

\bibitem[Ma09]{ma2009greed}
Bin Ma, \emph{Why greed works for shortest common superstring problem},
  Theoretical Computer Science \textbf{410} (2009), no.~51, 5374--5381.

\bibitem[MGMB07]{medvedev2007computability}
Paul Medvedev, Konstantinos Georgiou, Gene Myers, and Michael Brudno,
  \emph{Computability of models for sequence assembly}, Algorithms in
  Bioinformatics, Springer, 2007, pp.~289--301.

\bibitem[Muc13]{mucha2013lyndon}
Marcin Mucha, \emph{Lyndon words and short superstrings}, Proceedings of the
  Twenty-Fourth Annual ACM-SIAM Symposium on Discrete Algorithms, SIAM, 2013,
  pp.~958--972.

\bibitem[Ott99]{ott1999lower}
Sascha Ott, \emph{Lower bounds for approximating shortest superstrings over an
  alphabet of size 2}, Graph-theoretic concepts in computer science, Springer,
  1999, pp.~55--64.

\bibitem[RU81]{raiha1981shortest}
Kari-Jouko R{\"a}ih{\"a} and Esko Ukkonen, \emph{The shortest common
  supersequence problem over binary alphabet is np-complete}, Theoretical
  Computer Science \textbf{16} (1981), no.~2, 187--198.

\bibitem[Vas05]{vassilevska2005explicit}
Virginia Vassilevska, \emph{Explicit inapproximability bounds for the shortest
  superstring problem}, Mathematical Foundations of Computer Science 2005,
  Springer, 2005, pp.~793--800.

\end{thebibliography}

\end{document}